\documentclass[a4paper]{jpconf}
\usepackage{graphicx}

\usepackage{amsmath}
\usepackage{amsthm}
\usepackage{latexsym}
\usepackage{amsfonts}
\usepackage{amssymb}
\usepackage[utf8]{inputenc}

\usepackage{lmodern,bm}

\usepackage{bbm,dsfont}

\usepackage{graphicx}

\usepackage{verbatim}
\usepackage{cite}

\newtheorem*{theorem*}{Theorem}
\newtheorem*{proposition*}{Proposition}
\newtheorem*{remark*}{Remark}

\usepackage{color}

\usepackage[normalem]{ulem} 

\newcommand{\mods}[1]{\left \vert #1 \right \vert ^2}

\newcommand{\hi}{\mathcal{H}} 
\newcommand{\his}{\mathcal{H}_{\mathcal{S}}}
\newcommand{\hir}{\mathcal{H}_{\mathcal{R}}}
\newcommand{\hia}{\mathcal{H}_{\mathcal{A}}}
\newcommand{\hik}{\mathcal{K}} 

\newcommand{\Y}{\yen}

\newcommand{\lh}{\mathcal{L(H)}} 
\newcommand{\lhs}{\mathcal{L}(\mathcal{H}_{\mathcal{S}})} 

\newcommand{\lhr}{\mathcal{L}(\hir)} 
\newcommand{\lk}{\mathcal{L(K)}} 

\newcommand{\trh}{\mathcal{T(H)}} 
\newcommand{\sh}{\mathcal{S(H)}} 
\newcommand{\ip}[2]{\left\langle\,#1\,|\,#2\,\right\rangle} 

\newcommand{\id}{\mathbbm{1}} 

\newcommand{\Esf}{\mathsf{E}}

\newcommand{\Fsf}{\mathsf{F}}
\newcommand{\Zsf}{\mathsf{Z}}

\newcommand{\Sy}{\mathcal{S}}
\newcommand{\Ap}{\mathcal{A}}
\newcommand{\Bp}{\mathcal{B}}

\newcommand{\R}{\mathcal{R}}
\newcommand*\colvec[3][]{
    \begin{pmatrix}\ifx\relax#1\relax\else#1\\\fi#2\\#3\end{pmatrix}
}

\newcommand{\var}{\textrm{Var}} 

\newcommand{\E}{\mathsf{E}}
\newcommand{\F}{\mathsf{F}}
\newcommand{\Q}{\mathsf{Q}}

\newcommand{\U}{\mathcal{U}} 

\begin{document}
\title{A relational perspective on the Wigner-Araki-Ya- nase theorem}

\author{Leon Loveridge}

\address{Quantum Technology Group, Department of Science and Industry Systems, University of South-Eastern Norway, 3616 Kongsberg, Norway}

\ead{leon.d.loveridge@usn.no}

\begin{abstract}
We present a novel interpretation of the Wigner-Araki-Yanase (WAY) theorem based on a relational
view of quantum mechanics. Several models are analysed in detail, backed up by general considerations, which serve to illustrate that the moral of the WAY theorem may be that in the presence of symmetry, a measuring apparatus must fulfil the dual purpose of both reflecting the statistical behaviour of the system under investigation, and acting as a physical reference system serving to define those quantities which must be understood as relative. 
\end{abstract}

\emph{In this way quantum theory reminds us, as Bohr has put it, of the old wisdom that when searching for harmony in life one must never forget that in the drama of existence we are ourselves both players and spectators.}\hfill Werner Heisenberg \cite{heis}\\

\begin{center}{\bf Dedicated to the memory of Paul Busch} \end{center}

\section{Introduction}
In this paper we present an overview of some ideas developed in the decade $2008$--$2018$, which form part of a large and ongoing programme aimed at providing a novel understanding of quantum theory based on what one might call a ``symmetry-enforced relationalism". Some of this work was included in the author's PhD thesis, and was undertaken in collaboration with Takayuki Miyadera and the author's supervisor and friend Paul Busch, whose attention, guidance and encouragement made the work what it is, and is now so sorely missing it often feels impossible to proceed.

The work presented in this contribution to the volume dedicated to Paul remains unfinished, with both technical and interpretational work remaining. The idea is to combine the point of view that the world described by quantum mechanics is a deeply and  irreducibly relational one, with a technical result in the quantum theory of measurement called
the Wigner-Araki-Yanase (WAY) theorem, which imposes measurement limitations in the presence of conserved quantities, in the hope of both better understanding the message of this important theorem, and of gaining more general qualitative insight into the relationship between the formal aspects of quantum theory and the world that it is supposed to describe.

We will argue that the WAY theorem may be seen as an expression of the fact that the quantum theoretical formalism as it is traditionally used speaks not of quantum systems as they are,
but of the relation between a quantum system and a reference system which possesses certain properties
that we are accustomed to viewing as ``classical". This point of view, which emerged from our mutual discussion
and reflection over the years, was referred to by Paul with the acronym ``RRR", to stand for
\emph{radically relational reality}, the consequences of which we had only seriously begun to consider. Of course broadly speaking such an attitude is not new: Bohr emphasised very early
the epistemological view that quantum phenomena must be understood {\it only} in relation to the 
``whole experimental arrangement" \cite{bor2} and therefore pertain to both system and measuring apparatus combined -- one cannot speak of system without apparatus. There are many other forms of quantum ``relationalism" which are morally in harmony
with what we are advocating, for instance the relational quantum mechanics of Rovelli \cite{rov1}
and the perspectival approach of Bene and Dieks \cite{ben1,die1,die2}. However, the specific form
of relationalism presented here and in our other works (e.g., \cite{shortp,long}) is founded on symmetry in a fundamental way, and in that respect has more in common with the large body of work on \emph{quantum frames of reference} \cite{ak1,ak2,brs} -- a topic which also forms a major part of the current investigation.

After briefly introducing the rudiments of the quantum theory of measurement, we present (section \ref{sec:way})
the WAY theorem as it appears in \cite{LB2011,lb2,ldlthesis,blcon}, with some additional comments on developments in the years after those papers were published. The message of this work is principally that sharp observables 
not commuting with a conserved quantity cannot be measured exactly, 
but that good approximation is possible provided that the apparatus is able to be, and is, prepared
in a state with large fluctuation in the conserved quantity. In section \ref{sec:sym} we present two arguments for the adoption of the view that ``truly" observable quantities are invariant under
symmetry transformations. The first arises from the basic observation that this is standard in gauge theories, and in a Galilean spacetime context (for example) the possibility of measuring ``absolute" position or time (for instance) is at odds with Galilean relativity. The second is dynamical, 
and can be stated as a theorem in the quantum theory of measurement wherein the system under investigation possesses its own conserved quantity. In both cases, we find that reference frames/systems play a crucial role
in resolving the tension between the unobservability of ``absolute" (or non-invariant) quantities,
and the routine use of such quantities in the accurate description of the physical world. We find
that non-invariant quantities are theoretical shorthands for the true, relational/invariant quantities which relate system and reference. With this in hand, we will argue (section \ref{sec:relw}) that the WAY theorem
is an expression of the relativity, or relationalism, of physical quantities under symmetry,
where the measuring apparatus is also required to fulfil the role of reference. We conclude  
with a brief summary and discussion of further avenues of investigation for the RRR programme.

\section{Preliminary material}

\subsection{Basic framework}

To a quantum system $\Sy$ is assigned a complex, separable Hilbert space $\hi$ with inner product $\ip{\cdot}{\cdot}: \hi \times \hi \to \mathbb{C}$, linear in the second argument (and hence conjugate linear in the first). In case $\Sy$ consists of exactly two separate subsystems so that $\Sy$ can be viewed as a composite of $\Sy_1$ and $\Sy_2$, $\hi$ takes the form of the tensor product $\hi_1 \otimes \hi_2$. States of $\Sy$ are
represented by (and identified with) positive operators in the trace class $\trh$ of $\hi$ which have unit trace. This is a convex subset of $\trh$, which we will write as $\sh$. Pure states 
are singled out as the one dimensional projections, to be written as $P_{\varphi}$, with $\varphi$ a unit vector in $\hi$, and we will occasionally indulge the traditional laissez-faire inaccuracy and refer to 
$\varphi$ itself as a pure state of $\Sy$ (sometimes also to be called a \emph{vector state}).
 
Observables are identified with positive operator-valued measures ({\sc pom}s) $\Esf: \mathcal{F} \to \lh$, where $\mathcal{F}$ is a $\sigma$-algebra of subsets of some set $\Omega$ which represents the \emph{value space} of $\Esf$, i.e., the space of possible outcomes that may arise in a measurement of $\Esf$, and $\lh$ denotes the (algebra of) bounded linear operators in $\hi$. We recall that {\sc pom}s are positive in the sense that
all operators in the range of a {\sc pom} are positive (we write $\lh ^+$ for the positive elements of $\lh$), normalised ($\Esf (\Omega) = \id_{\hi}$) and have the property of being countably additive. The operators in the range of a {\sc pom} are called \emph{effects}.

That we use $\mathcal{F}$ rather than $\Omega$ owes itself to the probabilistic nature of the basic framework, to which we will shortly turn. \emph{Sharp} observables are singled out amongst all observables as the projection-valued measures ({\sc pvm}s), characterised as those {\sc pom}s $\Esf$ for which $\Esf(X)^2 = \Esf(X)$ for all $X \in \mathcal{F}$, or equivalently, $\Esf(X \cap Y) = \Esf(X)\Esf(Y)$ for all $X,Y \in \mathcal{F}$. In this case, and if $(\Omega, \mathcal{F}) = (\mathbb{R}, \mathcal{B}(\mathbb{R}))$ (the real numbers and the corresponding Borel sets) or subsets thereof, then $\Esf \equiv \Esf^A$ is the spectral measure of 
some unique self-adjoint operator $A \in \lh$, where $A$ and $\Esf^A$ are connected via the spectral integral $A = \int x \Esf^A(dx)$. We also refer to $A$ itself as an observable
when it is more convenient.

The state-observable pairings $(\rho, \Esf)$ thus defined capture the empirical content of the  theory via the ``Born rule" probability measures $X \mapsto \tr [\rho \Esf (X)]$, understood as representing the probability that if the system is prepared in a state $\rho$ and a measurement of $\Esf$ is performed, this measurement yields an outcome in the set $X$. Indeed, for an arbitrary mapping $\Esf: \mathcal{F} \to \lh ^+$, the assignment
$X \mapsto \tr [\rho \Esf (X)]$ being a probability measure for all states guarantees that $\Esf$ is a {\sc pom}, highlighting that this description of observables saturates the full probabilistic content of the Hilbert space framework. For a detailed account of the conceptual and mathematical  development of the use of {\sc poms} in the foundations of quantum mechanics and measurement we refer the reader to \cite{oqp,paulmmt,pbook}.

\subsection{Measurement}

Measurements of quantum objects inevitably involve an interaction of the object $\Sy$ under investigation with some second system, to be called the \emph{measuring apparatus} $\Ap$ which fulfils the dual role of being ``quantum enough" to interact properly with $\Sy$ (specifically, to become entangled with $\Sy$, which is indeed necessary for the measurement process---as information extraction---to take place -- see \cite{buschnie}) and ``classical enough" to produce measurement records. We will not discuss the problem of how such desiderata may be satisfied (a difficult question to answer, even in principle \cite{paulmmt,mitt}), but rather proceed to model such a situation operationally \cite{oqp}.

Given some observable $\Esf$ of $\Sy$ the question therefore naturally arises as to the existence
of a suitable apparatus (or \emph{probe}) system $\Ap$ which may be used to faithfully reproduce the statistics of $\Esf$, and therefore to function as a measurement of $\Esf$. Put more precisely, one wishes to specify an observable $\Zsf$ of $\Ap$ (to be called a \emph{pointer observable}), here assumed to be the spectral measure of a self-adjoint operator $Z$ (an assumption which will soon be justified) acting in $\hia$, a state $\rho_{\Ap} \in \mathcal{S}(\hia)$, and a suitable unitary mapping $U(\tau)$
acting in $\his \otimes \hia$ which serves to correlate $\Sy$ and $\Ap$ in a way which we will shortly spell out. The $\tau$ in appearing in $U(\tau)$ is a parameter related to the interaction strength and duration, which we will sometimes suppress, and suitable here means that we may arrange for the pre-interaction
statistics of $\Esf$ of $\Sy$, namely, the probability measures $X \mapsto \tr [\rho \Esf (X)]$ to be equal
to the post-interaction probabilities for the apparatus observable, i.e., 
$X \mapsto \tr \left(\Zsf (X)\tr _{\Sy} (U(\tau ) \rho \otimes \rho_{\Ap} U(\tau)^*) \right)$. 
This amounts to the statement of the \emph{probability reproducibility condition} (PRC) which serves as a minimal operational condition for the tuple $\mathcal{M}:= \langle \hia, U(\tau), Z, \rho_{\Ap} \rangle$, called \emph{measurement scheme} \cite{oqp}, to serve as a measurement of $\Esf$. This 
we write as
\begin{equation}\label{eq:prcm}
\tr [\rho \Esf (X)] = \tr [U(\tau)\rho \otimes \rho _{\Ap}U(\tau)^* \id \otimes \Zsf(X)].
\end{equation}

The existence of such an apparatus which satisfies the above condition is established by two central mathematical results. The first is \emph{Naimark's dilation theorem}, which states that given a {\sc pom} $\Esf:\mathcal{F}\to \lh$ there exists
a Hilbert space $\hik$, a {\sc pvm} $\Fsf: \mathcal{F} \to \hik$ and an isometry $V:\hi \to \hik$ for which 
\begin{equation}
\Esf(X)=V^*\F(X)V.
\end{equation}
The second, which builds upon this, is found in the remarkable contribution of Ozawa \cite{ozcont}, who showed that $\Esf$ always admits a \emph{measurement dilation}---a Naimark dilation of a specific form---comprising a Hilbert space $\hia$ for which $\hik = \hi \otimes \hia$, a fixed unit vector $\phi \in \hia$, and $V$ which can be chosen to be of the form
$V=UV_{\phi}$ with $U$ a unitary mapping $\hik \to \hik$ and $V_{\phi}: \hi \to \hik$ the isometric embedding $V_{\phi}(\varphi)=\varphi \otimes \phi$. Moreover, it can be arranged so that $\F(X) = \id \otimes \Zsf(X)$ where $\Zsf$ is the spectral measure of a self-adjoint operator $Z$. Thus we may write the probability reproducibility condition in this case (dropping the $\tau$ for now) and for pure states as
\begin{equation}\label{eq:pr}
\ip{\varphi}{\Esf(X) \varphi} = \ip{U(\varphi \otimes \phi)}{\id \otimes \Zsf (X) U(\varphi \otimes \phi)},
\end{equation}
stipulated to hold for all unit vectors $\varphi \in \hi$ and $X \in \mathcal{F}$.  This states that the probabilities for the pointer observable
$\Zsf$ in the vector state $U \varphi \otimes \phi$ perfectly reproduce those of $\E$ in the vector state $\varphi$,
and that therefore the measuring process for $\Zsf$ (and thus for $Z$) in the final state can be understood 
as a measuring process for $\E$ in the initial state. There are two conceptually distinct readings
of the PRC \eqref{eq:pr}. One is that the right hand side (i.e., $U, Z$ and $\phi$, and thus implicitly also $\hia$) is chosen so as to be a measurement (scheme) for some fixed $\Esf$ of interest to an experimenter. The other is to first fix $\mathcal{M}$, which reflects the common situation in empirical practice, wherein only certain apparatus preparations, dynamics, {\it etc.}, are available. Then the right hand side of \eqref{eq:pr} uniquely identifies the observable $\Esf$ actually measured by this scheme, which may or may not function as a good measurement of the observable of interest. 

Note that whilst the measurement outcome probabilities are fully characterised by the measured observable,
the measurement scheme provides a far more comprehensive account of the entire measurement process
and the information transfer between system and probe. The measurement scheme is not uniquely determined by the measured observable, and therefore it is often useful to consider the equivalence class of operationally indistinguishable schemes.

Occasionally it is necessary to include a measurable function $f: \Omega_{\Ap} \to \Omega$ in the definition of a measurement scheme in order to account for the measured observable and the pointer having different value spaces or scales; in this case equation \eqref{eq:pr} becomes 
\begin{equation}\label{eq:sca}
\ip{U\varphi \otimes \phi}{\id \otimes \Zsf(f^{-1}(X))U\varphi \otimes \phi}=\ip{\varphi}
{\Esf(X) \varphi}.
\end{equation} 

The second point of view mentioned above is perhaps more clearly expressed at the operator level by the equation
\begin{equation}\label{prc0}
\Esf(X)=V_{\phi}^*U^* (\id \otimes \Zsf(X))UV_{\phi};
\end{equation}
the mapping $\Gamma_{\phi} (\cdot) \equiv V_{\phi}^* (\cdot) V_{\phi} : \lk \to \lh$  is a (normal) conditional expectation and can be viewed
as a ``restriction" mapping on observables, arising from fixing a state of $\hia$. The observable defined through restriction may not
be equal to the ``target" observable which we actually wish to measure; if not, then statistical techniques are required in order to quantify the discrepancy between the observable we wish to measure (which will often be sharp) and the actually measured observable $\Gamma_{\phi}(U^*\id \otimes \Zsf(X))U)$, which may be unsharp and function as an unsharp approximator for $\Esf$.
Writing $Z(\tau) = U^*\id \otimes \Zsf(X))U$ highlights the fact that the measured observable
$\Esf$ is then just the restriction of the time-evolved pointer observable. 

The restriction map is not limited to describing measurement processes, can be naturally extended to mixtures,
and represents a natural means by which to describe a subsystem ``as though" it were being considered in isolation. There are various equivalent definitions which will be employed.
 It is the dual of the isometric embedding $\mathcal{S}(\hi) \to \mathcal{S}(\hi \otimes \hia)$ defined by
$\rho \mapsto \rho \otimes \sigma$. It is the continuous linear extension of the map
$A \otimes B \mapsto A \tr [B \sigma]$, and finally, it may be defined through the formula
\begin{equation}\label{eq:red}
\tr [\rho \Gamma_{\sigma}(\Lambda)] = \tr [\rho \otimes \sigma \Lambda],~\text{for all}~\rho \in
\mathcal{S}(\hi) ~\text{and} ~\Lambda \in \mathcal{L}(\hi \otimes \hia).
\end{equation}
For a vector state $\phi$ we continue to write $\Gamma_{\phi}$ rather than $\Gamma_{P_{\phi}}$.
In this work, we view the restriction as an \emph{externalisation} of the second system, i.e., a way to describe the subsystem contingent upon some preparation of the other system, which is being suppressed in the theoretical description but nevertheless functions crucially in the ascription of properties of $\Sy$ which pertain also to the state of $\Ap$.

We conclude this section by noting that the \emph{repeatability} of a measurement, which is often assumed to comprise part of the definition (of measurement) in standard textbook treatments, is not logically implied by the probability reproducibility condition and must therefore be viewed as an additional stipulation which a measurement may or may not satisfy. We recall that repeatable measurements are defined
as those for which the same outcome is guaranteed (with probability $1$) to be realised upon immediate repetition, i.e.,
\begin{equation}\label{eq:rep}
 \ip{U(\varphi \otimes \phi)}{\Esf(X) \otimes \Zsf (X) U(\varphi \otimes \phi)}= \ip{\varphi}{\Esf(X) \varphi},
\end{equation}
holding for all $\varphi$ and $X$.
Indeed, with the broader definition of observable than is traditional (i.e., the representation by {\sc pom}s), it is possible that no repeatable measurements exist at all. This is the case if, for example, some of the effects in the range of an observable do not have $1$ as an eigenvalue. Moreover, as shown by Ozawa in his 1984 contribution \cite{ozcont}, observables with continuous spectrum do not admit \emph{any} repeatable measurements, and therefore the extension of the notion of measurement to include those which are not repeatable is a crucial ingredient in the description
of physical experiments.

Historically, the question of whether all observables can, at least theoretically, be measured, perhaps also repeatably, was raised for the first time some twenty years after the landmark text of von Neumann \cite{vN}, which takes as a working hypothesis that all self-adjoint operators correspond to observable quantities. The Wigner-Araki-Yanase theorem rules this out, and is the subject of the next section.

\section{The Wigner-Araki-Yanase theorem}\label{sec:way}

\subsection{Background}
The framework presented above demonstrates that, in principle, any observable can be measured 
by specifying an appropriate measurement scheme. However, it may not be practicable to arrange
a suitable scheme for measuring arbitrary observables (for instance, it may be difficult to prepare the required $\phi$ in practice) and, {\it a fortiori}, there may be {\it fundamental} restrictions to the realizability of the elements of $\mathcal{M}$ required for a measurement (scheme) for $\Esf$, arising for instance from their incompatibility with other basic principles (see, e.g., \cite{fewster} for a recent example in quantum field theory). 

Two such proposed fundamental restrictions were raised in 1952, with Wigner involved in both.
One is the limitation imposed by a \emph{superselection rule} \cite{www}, which reduces the notion
of observable to include only those self-adjoint operators which sit in the commutant of some particular superselection
``charge" or, equivalently, prevents the preparation of coherent superpositions of eigenstates
of differing superselection charge. This subject has its own long and controversial history; see \cite{long}
and references therein for a recent review and analysis of this story, to aspects of which we will also return
in the present manuscript.

The other 1952 paper on fundamental obstructions to measurement was authored solely by Wigner 
(see \cite{wig1} for the original article, and \cite{buschtrans} for a facsimile-style translation into English), and investigated the possibility that conservation laws may interfere with the existence of measurement schemes satisfying the PRC. 

Specifically, Wigner observed that 
a simple \emph{repeatable} measurement, given as a specified state evolution resulting in a transformation of von Neumann/L\"{u}ders type, of the $x$-component of spin of a spin-$\frac{1}{2}$ particle conflicts with angular momentum conservation. We note four key observations made by Wigner in \cite{wig1}: 1) That the specified model of an $S_x$ measurement violates the conservation of $S_z \otimes \id +\id \otimes L_z$, the latter denoting the ($z$-component of) the angular momentum of the measuring apparatus, 2) That a modification of the dynamics so as to incorporate an ``error" state of the apparatus allows the conservation law to be satisfied, but the accuracy and repeatability of the measurement are lost, 3) Good accuracy and repeatability properties can be recovered to an arbitrarily high degree
if the initial state of the apparatus has a large uncertainty with respect to $L_z$, which can be viewed as a ``size" requirement on the apparatus and 4) That
if the repeatability condition is dropped completely, then perfectly accurate measurements are possible, even under the constraint of the conservation law, without any size condition on the apparatus. 

Regarding the fourth point, Wigner made the following observation (see also \cite{yanase}):
the reason for the measurement limitation is the noncommutativity between the spin
components of the system under investigation, {\it viz} $[S_x,S_z]\neq 0$. Dropping the repeatability condition in his model inevitably leads to a state evolution in which the final pointer states
are eigenstates of a quantity not commuting with $L_z$. Hence the question arises of how to measure the pointer.

The topic of conservation-law-induced limitations to quantum measurements  has been revisited in the literature until the present day. The main advancement after Wigner's paper was the 1960 contribution of Araki and Yanase \cite{ay1},
in which Wigner's example was shown to be an instance of a theorem about quantum measurements in general, wherein points 1) to 4) above are all seen to generalise in a natural way. See also \cite{LB2011} for a detailed discussion and reconstruction of Wigner's early work, a generalisation of the main result of \cite{ay1} which will play an important role in this manuscript, and a critical examination of the evolution of the WAY theorem up to 2011.

We now present the Wigner-Araki-Yanase theorem. The setting is that there is a system $\Sy$, an apparatus $\Ap$, and a quantity $L = L_{\Sy} \otimes \id + \id \otimes L_{\Ap} \equiv L_{\Sy} +L_{\Ap}$ which is \emph{conserved} in the sense that $[U,L]=0$. Following Ozawa \cite{oz02},
we shall call the condition that $[Z,L_{\Ap}]=0$ (which may or may not be satisfied) the 
\emph{Yanase condition}.

\subsection{WAY theorem for discrete observables}

The WAY theorem may be stated in the quantum theory of measurement as follows.

\begin{theorem*} 
Let $\mathcal{M} := \left\langle \mathcal{H}_{\Ap}, U, \phi, Z, \right\rangle$ be a measurement scheme for a discrete-spectrum self-adjoint operator $A$ on $\mathcal{H_S}$, and let $L_{\Sy}$ and $L_{\Ap}$ be bounded self-adjoint operators on $\mathcal{H_S}$ and $\mathcal{H_A}$, respectively, such that $[U, L_{\Sy} +L_{\Ap}] = 0$. Assume that $\mathcal{M}$ is repeatable or satisfies the Yanase condition. Then $[A,L_{\Sy}]=0$.
\end{theorem*}

See \cite{LB2011} for a proof and discussion of some generalisations. The most important one for 
us is that the unboundedness of $L_{\Ap}$ can be easily accommodated by taking the exponential 
of $L_{\Ap}$ (which gives a bounded operator) and running the argument as before -- see footnote 4 in \cite{ay1}. The theorem entails that if $[A,L_{\Sy}] \neq 0$, then any measurement of $A$ is necessarily non-repeatable,
and must violate the Yanase condition. The \emph{SWAP map}, defined when $\his \cong \mathcal{H_A}$ by $\varphi \otimes \phi \mapsto \phi \otimes \varphi$, is perhaps the simplest example of such a realization \cite{LB2011}. 

The Yanase condition was not known to be a condition in the statement of the WAY theorem until fairly recently (its first appearance in that capacity is in \cite{belt}; see also \cite{LB2011}) and is worth dwelling on further, as it will play an important role in the sequel. From a physical perspective, such a requirement appears to be natural: applying the WAY theorem to the pointer (imagine measuring the pointer) prohibits repeatable measurements of $Z$, which surely nullifies its purpose as a stable
pointer that can be repeatedly checked by an experimenter. 

The statement of the WAY theorem does not rule out approximate
measurements of $A$, understood in the sense that there may be some unsharp observable $\Esf$
whose statistics are close to those of (the spectral measure of) $A$ and which may nevertheless \emph{not} commute with $L_{\Sy}$. Moreover, the WAY theorem as stated also does not prohibit such approximate measurements
from having approximate repeatability properties. Such a possibility can indeed be realised, and can be viewed as a positive counterpart to the strict impossibility verdict, in line with Wigner's earlier observations. Even if $[A,L_{\Sy}]\neq 0$, approximate measurements of $A$, with approximate repeatability properties, are feasible, with increasingly good approximation properties becoming possible with a large ``spread" of the apparatus part of the conserved quantity in the initial state of the apparatus. Large spread in those early works corresponded to a large number of components of the given state when written as a linear combination of eigenstates of the conserved quantity with different eigenvalue. Other measures are also possible---variance for instance---and as argued in \cite{LB2011}, such a spread, appropriately quantified, is \emph{necessary} for good measurements. 

The case of position measurements obeying momentum conservation is not covered by the WAY theorem as presented above, yet is of clear physical importance. This case also carries an interesting message regarding
the relational perspective on the WAY theorem, with which we will eventually make contact, and therefore we will discuss the position-momentum case in some detail next.

\subsection{Position and momentum}

WAY-type constraints persist also in the continuous-variable and unbounded case, not covered by the original proofs, as reported for the first time in \cite{lb2}; we will be interested here in position (measurements) and momentum (conservation). We present a model due to Ozawa (see \cite{Ozawa} and \cite{lb2,ldlthesis} for further analysis of this model) which demonstrates the presence of WAY-type limitations
for position measurements, and contrast this with a second model \cite{oqp} which, although
also has momentum conservation built in, shows no such WAY-type behaviour. We then present 
a heuristic argument in favour of WAY-type constraints persisting in general for $Q$ and $P$.

\subsubsection{Two models} 

Ozawa's measurement scheme is characterised by the following unitary coupling 
acting in $\his \otimes \hir \otimes \mathcal{H}_{\Ap}\otimes \mathcal{H}_{\Bp}$ (we omit the identity factors and tensor product symbols for notational clarity):
\begin{equation}
U=e^{i \frac{\lambda}{2}(Q-Q_{\mathcal{R}})(Q_\Ap - Q_\Bp)}.
\end{equation}
All Hilbert spaces are taken to be isomorphic to $L^2(\mathbb{R})$, understood as position representation spaces for $\his$ and $\hir$, and momentum representation spaces for 
$ \mathcal{H}_{\Ap}$ and $\mathcal{H}_{\Bp}$. We may therefore write $ \mathcal{H}_{\Ap} \otimes \mathcal{H}_{\Bp} \cong L^2 (\mathbb{R}^2)$ with $\mathbb{R}^2$ parametrised by
$(p_{\Ap},p_{\Bp})$. Then under the orthogonal transformation $(p_{\Ap},p_{\Bp}) \mapsto (u,v)= \frac{1}{\sqrt{2}}(p_{\Bp} - p_{\Ap},p_{\Bp} + p_{\Ap})$, we have $ \mathcal{H}_{\Ap} \otimes \mathcal{H}_{\Bp} \cong \tilde{\mathcal{H}}_{\Ap} \otimes \tilde{\mathcal{H}}_{\Bp} \cong L^2(\mathbb{R}^2)$ with coordinates $(u,v)$, which are spectral values of 
$P_\Bp-P_\Ap$ and $P_\Bp+P_\Ap$ respectively. Finally, the pointer observable is
chosen to be $Z = P_\Bp-P_\Ap$, which satisfies the Yanase condition (we ignore subtleties arising from domain questions). We continue to denote by $\Zsf$ the spectral measure of $Z$. Crucially, $U$ conserves the total momentum $P +P_{\R}+P_{\Ap}+P_{\mathcal{B}}$.

The scheme is constructed so as to measure (an unsharp approximation of) $Q$---the position observable of $\Sy$---and therefore the 
remaining systems $\R$, $\Ap$ and $\mathcal{B}$ are all treated as part of the apparatus (though the $\R$ system does not comprise part of the pointer and plays a role which we will investigate in more detail in due course). The apparatus is initially prepared in the vector state
$\phi \otimes \xi_a \otimes \xi_b \in \hir \otimes \tilde{\mathcal{H}}_{\Ap} \otimes \tilde{\mathcal{H}}_{\Bp}$, and the measured observable is extracted from the probability reproducibility condition (where we in this instance label identity factors to properly identify subsystems)
\begin{equation}
\ip{\varphi}{\Esf (X) \varphi} = \ip{U (\varphi \otimes \phi \otimes \xi_a \otimes \xi_b)}{\id _{\his} \otimes \id _{\hir} \otimes \Zsf (f^{-1}( X)) \otimes \id _{\mathcal{H}_{\Bp}} U (\varphi \otimes \phi \otimes \xi_a \otimes \xi_b)},
\end{equation}
and can be shown \cite{ldlthesis} to be of the form  
\begin{equation}
\Esf (X) =  (\chi _X * e^{(\lambda)})(Q),
\end{equation}
where $\chi _X$ denotes the characteristic set function, $*$ is convolution and the scaling function $f$ (cf.\ eq.~\eqref{eq:sca}) has been chosen as $f(x)=-(2/{\lambda})x$. Note that the $\xi_b$ plays no role at all. The {\sc pom} $\Esf$ represents a \emph{smeared position observable} \cite{oqp,pbook}, with $e^{(\lambda)}$ a density (or confidence function), itself also a convolution given by
\begin{equation}
e^{(\lambda)}(x) =  \bigl|{\phi}\bigr|^2*\bigl|{\xi_a ^{(\lambda)}}\bigr|^2(x),
\end{equation} 
where $\xi_a ^{(\lambda)}(s) = \sqrt{\lambda}\xi_a (\lambda s)$.
The smearing, 
or inaccuracy, is dictated
by the (spread of the) density function $e^{(\lambda)}$; the more tightly peaked  
 $e^{(\lambda)}$, as quantified for instance by the variance $\var(e^{(\lambda)})$, the better the measurement of $Q$. This is understood in the sense of statistical proximity between $\Q$ (the spectral measure of $Q$) and the actually measured observable $\Esf$. 
 
In the above case we find that $\var(e^{(\lambda)}) = \var \mods{\phi} + \frac{4}{\lambda ^2}\mods{\xi_a}$, which is strictly positive, and therefore $\mathcal{M}$ always realised an approximate determination of position. Since $\lambda$ can be made arbitrarily large by choosing a suitably large interaction strength,
$\var(e^{(\lambda)})$ is therefore bounded below by the spread (as variance) of the distribution 
$\mods{\phi}$, which can be made small only by a high degree of localisation with respect to position $Q_{\R}$ or, by the uncertainty relation, a large uncertainty/delocalisation/spread with respect to $P_{\R}$, providing the continuous-variable analogue of the positive part of WAY.

A better measure of spread than variance (or standard deviation) of a probability distribution is the \emph{overall width} $\mathcal{W}(e;1-\epsilon)$ at confidence level $1 - \epsilon$, defined as the smallest possible size of an interval $J$
for which $\int_J e(x) dx \geq 1-\epsilon$ (see \cite{pbook}). The overall width is finite for all $\epsilon > 0$, and the overall width of a convolution of probability measures is bounded below by the largest, i.e.,
$\mathcal{W}(e^{\lambda} ;1-\epsilon) \geq \max \{\mathcal{W}(\mods{\phi}; 1-\epsilon),\mathcal{W}(\frac{4}{\lambda ^2}\mods{\xi_a};1-\epsilon) \}$.
Hence, $\mathcal{W}(e^{\lambda} ;1-\epsilon) \geq \mathcal{W}(\mods{\phi}; 1-\epsilon)$. In other words, the quality of the approximation of $\Q$ by $\Esf$ is dictated precisely by the $Q_{\R}$ localisation of $\phi$ and is bounded from below by this quantity. We see therefore that
it is the second ($\R$) system which plays a vital role in the quantitative part of WAY here, and \emph{not} the $\Ap$ and $\mathcal{B}$ systems pertaining to the pointer. 

The Ozawa model can be contrasted with another model \cite{oqp} (again we simplify the notation by omitting tensor product and identity symbols) described by (with $\lambda > 1$)
\begin{equation}
U = \exp \left[ {-i\frac{\lambda}{2}\left((Q-Q_{\Ap})P_{\Ap} + P_{\Ap}(Q-Q_{\Ap}) \right)}\right],
\end{equation} 
with pointer $Q_{\Ap}$. This model manifestly conserves momentum, and clearly violates the Yanase condition. This model may be viewed as a kind of ``symmetrised version" of the position measurement of von Neumann as described in the final few pages of his classic text \cite{vN}.\footnote{Von Neumann's model is described by the unitary coupling $U(\lambda) = e^{i \lambda Q \otimes P_{\Ap}}$. As pointed out in \cite{buschnid}, von Neumann compared the
particle position and the pointer reading after the cessation of the interaction period, thereby
analysing the \emph{repeatability} properties of the scheme. Had he also investigated the correlation of the position \emph{before} the interaction to the pointer afterwards, i.e., 
performed a calculation associated with the PRC, he may well have discovered the representation of observables by {\sc poms} in 1932, around 40 years before their eventual introduction.} 

The measured observable is again extracted from the PRC by noting that the initial
state of the system and apparatus combined $\Psi _0 = \varphi \otimes \phi \equiv \varphi \phi \in L^2(\mathbb{R}) \otimes L^2(\mathbb{R})$ (with position-representation wavefunction $\Psi _0 (x,y)$) evolves to $\Psi_t (x,y) = e^{\frac{\lambda t}{2}}\varphi(x)\phi \left((1-e^{\lambda t})x + e^{\lambda t}y\right)$
and the measurement is completed at $t=1/2$. The scaling function in this model is
$f(x) = (1-e^{-\lambda})^{-1}x$, and the observable measured is again a smeared position
$\Esf(X) = (\chi_X * e^{(\lambda)})(Q)$, with $e^{(\lambda)}$ the confidence function  $e^{(\lambda)}(x) = (e^{\lambda} - 1)\vert{\phi (-x)(e^{\lambda} -1))}\vert ^2$ (recall that $\phi$ represents the initial state of $\Ap$.)

This model is therefore arbitrarily accurate for arbitrarily large interaction strength $\lambda$ (in the sense that the actually measured observable $\Esf(X) \equiv \chi _X * e^{(\lambda)}(Q)$ can be made as close as one likes to the sharp $\chi_X(Q) \equiv \mathsf{Q}(X)$), and as it turns out has arbitrarily good repeatability properties, both without any condition on $\phi$.\footnote{This model played the central role in providing an ``approximate solution" of an old question
of Shimony and Stein \cite{ssp} regarding the (im-) possibility of a two-valued (coarse-grained) position measurement respecting the conservation of momentum. The problem was posed in the research problems section of \emph{The American Mathematical Monthly}, which contained ``easily stated research problems dealing with notions ordinarily encountered in undergraduate mathematics." 
It reappeared in 2009 \cite{beq} as the first of a list of six unsolved ``bequest" problems of
 Shimony, based presumably on a 2006 talk of his at the Perimeter Institute. The question of whether the solution presented
in \cite{LB2011} (which is not in perfect accordance with the stated problem but is operationally indistinguishable from a positive answer to it) was---to the enormous regret of the author---never put to Abner.} Therefore, with the Yanase condition violated, arbitrarily accurate position measurements can be realised whilst respecting the conservation of momentum, and there is no WAY-type limitation in this model.\footnote{That the inaccuracy scales with $e^{-\lambda}$ in this model, in comparison to von Neumann's model which scales with $\lambda ^{-1}$ (and does not conserve momentum) was pointed out to be surprising by Paul in a private communication in 2009 -- surprising because one would have thought, on WAY-type grounds, that incorporating momentum conservation would make matters worse, not exponentially better. In fact, this observation led Paul to doubt his old result
in \cite{oqp} on the scaling behaviour in the momentum-conserving model, which in turn led to the author's first investigation into WAY-type limitations for continuous observables, from which a large body of work has followed. Of course, in 2009 we did not know about the role of the Yanase condition: Paul's old calculation was right, and by violating the Yanase condition good measurements are feasible without any constraint on the apparatus.} See \cite{ldlthesis, oqp, lb2, pbook} for further discussion. 

\subsubsection{A general inequality}\label{ssqi}
Finally we give an argument using an inequality due to Ozawa \cite{oz02}, based the 
concept of a \emph{noise operator} $N:=Z(\tau) - Q$. As an unbounded operator there are of course
sensitive issues here relating to domains, which we continue to surreptitiously ignore. The use of such a quantity has also been (correctly) criticised (see, for example, \cite{brmp}) in general, the reason being that the observables being compared (in this case, $Z(\tau)$ and $Q \otimes \id$ ) do not necessarily commute, therefore requiring separate experimental procedures for the following estimate. Hence what comes next must be taken as heuristic and not operationally meaningful in general.
Setting $\epsilon ^2 := \sup_{\varphi} \epsilon (\varphi)^2:=\ip{\varphi \otimes \phi}{N^2 \varphi \otimes \phi} \geq (\Delta N)^2$,
the uncertainty relation then gives
\begin{equation}\label{eq:gto}
\epsilon ^2 \geq \epsilon (\varphi)^2 \geq \frac{1}{4} \frac{\mods{\langle [Z(\tau)-Q, P + P_{\Ap}]\rangle}}{(\Delta P_T)^2};
\end{equation}
the supremum is taken over unit vectors and is assumed to be finite across all states, with
$\epsilon$ then understood as a global measure of ``error", and $(\Delta P_T)^2 = (\Delta _{\varphi}(P))^2 + (\Delta _\phi (P_{\Ap}))^2$. The measurement is deemed accurate exactly when $\epsilon = 0$. From here we observe that i) if the right hand side is non-zero for some system states, there is a measurement limitation, ii) that if the numerator 
is non-vanishing, the only way to make the lower bound to the error small, without reference to the properties of the system, is by making  $(\Delta _\phi (P_{\Ap}))$ large, and (iii) that the numerator $[Z(\tau) - Q,P+P_{\Ap}]$ vanishes if $[Z,P_{\Ap}]=i$, which is a statement of the violation of the Yanase condition and is exactly the situation encountered in the second model above. Therefore, the $Q,P$ case in many ways reflects the discrete/bounded scenario (subject to some caveats; see \cite{LB2011}).

We finally mention that Eq. \eqref{eq:gto} generalises, subject to the same objections, to arbitrary quantities, giving
rise to an inequality of the following form:
\begin{equation}\label{eq:tot}
\epsilon ^2 \geq \epsilon (\varphi)^2 \geq \frac{1}{4} \frac{\mods{\langle [Z(\tau)-A, L_{\Sy} + L_{\Ap}]\rangle}}{(\Delta L)^2},
\end{equation}
where again the measurement is accurate exactly when $\epsilon = 0$, with a broadly similar conclusion to above.

\subsection{Further developments}

From 2011 until the present, there have been several interesting developments in relation
to the WAY theorem. Technical progress for instance includes the generalisation of the impossibility part due to Tukiainen \cite{mikko}, where additivity and even conservation are dropped but the theorem nevertheless retains its main flavour, and the generalisation due to \L{}uczak \cite{luc} based on von Neumann algebras and instruments. The WAY theorem has inspired further analysis of the case of energy conservation in \cite{pop1}, and has even found applications in the rapidly emerging field of quantum thermodynamics; see, e.g., \cite{hamdog,ham2}.
 
Perhaps the main lines of development have come via the general perspective that conservation 
laws are instances of symmetry constraints. The 2013 contribution of Ahmadi, Jennings and Rudolph \cite{ahmadi1} proposes that the WAY theorem may be understood in the context of \emph{resource theory}---a subject which is by now a research programme in its own right (see \cite{gour1} for a recent review)---in which there is collection of preparations, transformations, measurements which
are ``permissible" from the point of view of some constraint, and everything else is viewed 
as a \emph{resource}. In this instance, permissible states and observables are symmetry-invariant,
and the large spread in the conserved quantity of the initial state of the apparatus is then the asymmetric resource, which can be appropriately quantified. 

There remains a large variety of views on the ``true meaning" of the WAY theorem. These include
the resource-theoretic perspective mentioned above, the more informational point of view as 
first proposed in \cite{spek1}, and the informal connection to measurement disturbance \cite{pop1}, to name a few.
However, it has been intimated by many that the WAY theorem is not intuitive, or at least that the conceptual message is unclear. Here we put forward the idea that in certain cases, namely when the 
conserved quantity has a (covariant) conjugate quantity, the WAY theorem is a statement of
the limitation on the possibility of ``absolute" quantities being well approximated by
invariant quantities of system-plus-reference, wherein the reference is also playing the role of measuring apparatus. We begin with a general analysis of observability under symmetry, putting WAY to one side for a moment.

\section{Symmetry and observability}\label{sec:sym}

We investigate the relationship between three important notions: symmetry, the natural relativity of certain physical quantities, and the requirement of reference frames or, more properly, reference systems, for the relativity and symmetry to manifest.

\subsection{Observables as invariants}

There are strong reasons to adopt the view that in the presence of 
symmetry, \emph{only} those quantities (represented as self-adjoint operators or {\sc pom}s) which are invariant under relevant symmetry transformations should be considered to have the status of (being) observable. This is standard in the case of gauge
symmetry (e.g., \cite{haag}), and can be argued more generally. For example, ``absolute" position or angle or time are not observable classically, 
and are manifestly non-invariant quantities under the relevant transformations. On the other hand, 
\emph{relative} position, angle and time (see \cite{time} for an analysis of the case of time in quantum theory), are invariant and observable. The non-observability of ``absolute" quantities is in accordance with fundamental (in this case, space-time) symmetries. Therefore we see that many
standard theoretical notions such as position acquire their meaning \emph{only} in relation
to other physical objects. Classically, the observable relative position can be replaced with the 
non-observable absolute position by suppressing the (variables of the) second system which functions as a \emph{reference} with no difficulty. The question
of whether this is possible in quantum mechanics is one of the main questions addressed 
in \cite{long}, to which we will repeatedly refer. 

The situation in quantum theory is complicated by the fact that the reference system
should (must?) be treated as a quantum system, and is therefore subject to properties
held by quantum systems, such as indeterminacy, unsharpness, and entanglement, among others \cite{long}. On the other hand,
physicists routinely use ``absolute" quantities in the description of microscopic physical systems, seemingly without contradiction or even imprecision.

There is another reason for demanding the invariance of observables, which comes from conservation
and from the quantum theory of measurement. Suppose first that $\Sy$ itself has a 
conserved quantity, which we continue to write as $L_{\Sy}$. Then the following theorem
shows that \emph{any} observable (sharp or unsharp) must commute with $L_{\Sy}$ in order to be measured (or, we argue, deemed truly (an) ``observable") -- a situation
which was dubbed the ``strong WAY theorem" in \cite{long} due to the stronger form of the conservation law than in the original WAY theorem (where the conservation is stipulated only for system and apparatus combined). 
\begin{proposition*}
Consider a measurement scheme $\mathcal{M}$ for $\Esf$, and write
$U_{\Sy}(\ell) = e^{i \ell L_{\Sy}}$ for some real parameter $\ell$ and suppose that $[U, U_{\Sy}(\ell)]=0$ for all $\ell$ \emph{(}this is equivalent to the conservation of $L_{\Sy}$\emph{)}. Then 
\begin{equation}
U_{\Sy}(\ell)^*\Esf (X) U_{\Sy}(\ell) = \Esf (X),
\end{equation}
for all $\ell$ and $X$. 
\end{proposition*}
\begin{proof}
Without loss of generality suppose that $f = 1$.  Then the PRC is
\begin{equation}
\ip{U \varphi \otimes \phi}{\id \otimes \Zsf (X) U \varphi \otimes \phi} = \ip {\varphi}{\Esf (X) \varphi};
\end{equation}
substituting $\varphi$ for $U_{\Sy}(\ell)\varphi$ leaves the left hand side of the above unchanged, and hence the right hand side is unchanged as well, proving the result. 
\end{proof}
This holds for general observables and conserved quantities; we make no assumptions about their spectra
or even their boundedness. The result imposes stronger constraints than the WAY theorem, ruling out any measurements
of any (sharp or unsharp) quantities not commuting with a conserved quantity. Note that this is strongly reminiscent of the situation encountered when there is a superselection rule.

The two perspectives (on the unobservability of certain theoretical quantities) are of a different
source and nature. However, in both cases there is a tension between the apparently defensible
viewpoint that ``absolute" quantities such as position do not represent observable quantities,
and the routine use of such quantities in the (accurate) description of the physical world. And in
both cases, we will argue, the resolution is found via the introduction of quantum frames
of reference, which we now discuss.

\subsection{Quantum reference systems and relationalism}

That certain physical quantities obtain their meaning only in relation to a standard
of reference is by now presumably uncontroversial. However, there are differing viewpoints on whether such references are to be understood as abstract ``coordinate systems" or as physical
objects. Our position is that the latter should be the starting point: that relations between
physical entities is primary, and that abstract coordinatisation is an idealisation that must reflect actual physical possibility. For instance, in a spatial Newtonian setting, positions of objects
are understood as relative to other objects, these objects being localised in space and any of which can serve as a reference. Once it is established that such physical references exist and function as they do, the process of abstraction to ``mathematical" references can take place and
these references are, in the setting just discussed, systems of Cartesian coordinates. 

As Einstein has stressed in the classical relativistic setting, ``Every description of events in space involves the use of a rigid body to which such events have to be referred. The resulting relationship takes for granted that the laws of Euclidean geometry hold for ``distances", the ``distance" being represented physically by means of the convention of two marks on a rigid
body" \cite{ein}. That 
the step of abstraction is possible appears to be so natural in classical physics that it is seldom considered. Einstein's remark encapsulates the operational philosophy: the relative position is understood as an observable quantity and one must therefore provide a physical means by which to measure it. The theoretical use of the absolute position is only justified as a mathematical ``shorthand", if indeed it is justified.

If quantum theory is understood to be universally applicable, the reference bodies must
be quantum physical systems in the world, governed by the laws of quantum theory. The manoeuvre 
by which the reference system is externalised/abstracted must then be carefully examined and justified. 
We discuss this possibility now, referring again to \cite{long} for the full treatment, and 
\cite{shortp} for the condensed version.

The observation that symmetry (invariance) and the relativity of quantities seem to come together
is interesting and worth investigating in its own right. We proceed under the point of view that
the symmetry is somehow ``prior", if only because one can formulate the question we wish to discuss
in a clean way, and because it may be more general. 

\subsection{Relativisation}

Suppose we have some group $G$ (with some properties which we shall soon discuss)
acting both in the value space $\Omega$ (with $\sigma$-algebra $\mathcal{F}$) of some observable 
$\Fsf$ of $\R$ and via strongly continuous unitary representations $U_{\Sy}$ and $U_{\R}$
in the corresponding spaces, and suppose moreover that $\Fsf$ is \emph{covariant} under the action of $G$: 
\begin{equation}
U_{\R}(g) \Fsf (X) U_{\R}(g)^* = \Fsf (g.X)~\text{for all}~g\in G,~X \in \mathcal{F}.
\end{equation}
If we can identify $\Omega$ and $G$ then the following map may be used to construct invariant quantities in $\Sy + \R$ from arbitrary ones of $\Sy$:\footnote{The unusual notation owes itself 
to a suggestion of Busch, circa 2011, that since $\Y$ generalises and makes rigorous the map $\$$ appearing in \cite{brs}, and since the map may be accredited to Miyadera who introduced it in a visit to York in 2011, a currency exchange may be appropriate.}
\begin{equation}\label{eq:dyen}
\Y : \lhs \to \mathcal{L}(\his \otimes \hir); ~~~~A \mapsto \int_{G}U_{\Sy}(g)AU_{\Sy}(g)^* \otimes \Fsf (dg).
\end{equation}
This extends naturally to {\sc pom}s by setting $(\Y \circ \Esf)(X):= \Y (\Esf (X))$.
Equation \eqref{eq:dyen} represents the integral of an operator-valued function with respect to an operator measure; see \cite{long} for the construction. We note that a sufficient condition for the integral
to be well defined is that $G$ is compact and metrisable if $\hir$ is finite \cite{davies}, and that $G$ is abelian and second countable in the general case \cite{long}. The mapping $\Y$ can be seen to enjoy a host of agreeable 
qualities: it is (completely) positive, normal, unital, $^*$-preserving, and moreover
an algebraic $^*$-homomorphism if $\Fsf$ is projection-valued. Crucially, as is easily confirmed,
$\Y(A)$ is invariant under $\Y(A) \mapsto U_{\Sy}(g)\otimes U_{\R}(g)\Y(A) U_{\Sy}(g)^*\otimes U_{\R}(g)^* \equiv \alpha_g(A)$.

The $\Y$ map therefore ascribes to each {\sc pom} of $\Sy$ a new observable of $\Sy + \R$ which is invariant under the given unitary representation of $G$. As will be exemplified below, we view this procedure as the explicit incorporation of a system of reference into the
theoretical description.

We give some examples. Suppose $G = (\mathbb{R}, +)$ (the additive group on the real line; we henceforth denote this group by $\mathbb{R}$) with strongly continuous unitary representations
$U_{\Sy}$ and $\U_{\R}$ acting in $\his$ and $\hir$, respectively. We write $U_{\Sy}(x) = e^{iP_{\Sy}x}$ and $U_{\R}(x) = e^{iP_{\R}x}$, where $P_{\Sy}$ and $P_{\R}$ are shift-generators, interpreted physically as momenta. Choosing $\Fsf$ to be the sharp position observable $\Fsf = \Esf^{Q_{\R}}$, we find that
\begin{equation}
(\Y \circ \Esf^{Q_{\Sy}})(X) = \int_{\mathbb{R}} \Esf^{Q_{\Sy}}(X+x)\otimes \Esf^{Q_{\R}}(dx) = \int_{\mathbb{R}} \int_{\mathbb{R}} \chi_X(x^{\prime}-x)\Esf^{Q_{\Sy}}(dx^{\prime})\otimes \Esf^{Q_{\R}}(dx)  = \Esf^{Q_{\Sy}-Q_{\R}}(X). 
\end{equation} 
Since $\Esf^{Q_{\Sy}-Q_{\R}}$ is the spectral measure of the self-adjoint relative position
observable $Q_{\Sy}-Q_{\R}$, we observe that $\Y$ in this case has the effect of ``relativising" the absolute position $Q_{\Sy}$, yielding $Q_{\Sy}-Q_{\R}$.

We may also consider the case of the circle $S^1$, defined by the interval $(-\pi, \pi]$ with the end points identified and addition mod $2\pi$. As a group this is then isomorphic to $U(1)$ and in the example we are about to give, $\hir \cong L^2(S^1)$. Choosing $\Fsf = \Esf^{\Phi_{\R}}$ to 
be the spectral measure of the azimuthal angle operator $\Phi$ conjugate to (say) the $z$-component 
$L_z$ of angular momentum, we have 
\begin{equation}
(\Y \circ \Esf^{\Phi_{\Sy}})(X) = \Esf^{\Phi_{\Sy} - \Phi_{\R}}(X);
\end{equation}
here, $\Esf^{\Phi_{\Sy} - \Phi_{\R}}$ is the spectral measure of the relative angle 
$\Phi_{\Sy} - \Phi_{\R}$, and therefore we may write $\Y(\Phi_{\Sy}) = \Phi_{\Sy} - \Phi_{\R}$.

We also mention the physically important quantity of phase conjugate to number, which generalises the angle/angular momentum pair; we write
$\Fsf:\mathcal{B}(S^1) \to \lhr$ for such a quantity. Conjugacy here means that phase shifts are generated by
number operators, i.e., $e^{iN_{\R}\theta}\Fsf(X)e^{-iN_{\R}\theta}=\Fsf(X+\theta)$, where addition is understood modulo $2 \pi$. The relativisation $\Y$ gives rise to a relative phase observable \cite{rp1,rp2}, as shown in \cite{long}. 
 If the spectrum of $N_{\R}$ is bounded from below, as in the harmonic
oscillator for example, $\Fsf$ cannot be sharp. Note that if $N_{\R}$ can be defined in a Hilbert space
of low dimension, in which case $\Fsf$ is highly unsharp -- the spin phase \cite{oqp} acting in $\mathbb{C}^2$ is such an example. 

Therefore we observe that the effect of relativising ``absolute" position and phase (with angle as the special case where $N_{\R}$ is two-sided-unbounded) is to
produce relative position and phase, as one might expect from a relativising map.
However, $\Y$ is not confined to acting only on {\sc pom}s which have familiar relative counterparts as (the spectral measures of) difference operators, or already established relative versions such as relative phase, and is quite general. We therefore view any observable obtained under $\Y$ as relative, or \emph{relational}, and, even more generally, \emph{any} invariant 
quantity of $\Sy + \R$.\footnote{It is known that in general $\Y$ is not surjective on the invariant part of $\mathcal{L}(\his \otimes \hir)$ \cite{LLandJW}.} Again,
we refer to \cite{ldlthesis,long} for further discussion and examples.

\subsection{Localisation and approximation}

One of the main results of \cite{long} is in regard to the possibility of making
$A$ and $\Y(A)$ statistically close (or more generally with $A$ replaced by a {\sc pom}). The interest in this question stems directly from the fact that in standard applications of quantum mechanics, 
``absolute" quantities such as position and phase are empirically adequate 
theoretical objects, yet, as we know, cannot be represented in physical reality because they are not invariant. The question therefore arises as to whether, and under what conditions, the unobservable ``absolute" and the truly observable relative quantities can at least approximately agree.

Using the $\Y$ map we find that good approximation is possible \emph{provided} the relativising quantity $\Fsf$ possesses a technical localisability property
called the norm-$1$ property (e.g., \cite{no1}), namely, that for any $X$ for which $\Fsf(X) \neq 0$, there
exists a sequence of unit vectors $(\phi_n) \subset \hir$ for which $\lim _{n \to \infty}\ip{\phi_n}{\Fsf (X) \phi_n}=1$. We note that this is satisfied for all {\sc pvm}s.  
Then we have 
\begin{theorem*}
Let $G$ be $\mathbb{R}$ or $S^1$ and $\Gamma$ be defined as in equation \eqref{eq:red}. There is a sequence of unit vectors $(\phi_n) \subset \hir$ for which
\begin{equation}
\lim_{n \to \infty}(\Gamma_{\phi_{n}}\circ \Y)(A) = A,
\end{equation}
where the convergence is understood weakly, i.e., in the topology of pointwise convergence of expectation values.
\end{theorem*}
The essence of the proof (see \cite{long} for the details) lies in the observation that
\begin{equation}
(\Gamma_{\phi_{n}}\circ \Y)(A) = \int _{G} U(g)^*AU(g) \mu _{\phi_n}^{\Fsf}(dg),
\end{equation}
 and that if $\Fsf$ satisfies the norm-1 property then we can make $\mu _{\phi_n}^{\Fsf}$ as concentrated as we like, i.e., approximately contained in any measurable set. Fixing this to be
 centred on $0$ (the identity of $G$) then yields the result. The localisation of the state
 at the identity of $G$ acts as a sort of ``encoding" of this element of $G$ in $\hir$ and serves as the zero-reference. 
 
There are various further observations to be made. The noise operator, subject to the already-discussed 
substantial caveat of not being operationally meaningful in general, can be used to compare 
$\Y(A)$ and $A$ for representations of $G$ generated by the additive quantity
$L = L_{\Sy} + L_{\R}$ on $\his \otimes \hir$, with state $\varphi \otimes \xi$ (see \ref{ssqi}):

\begin{equation}\label{eq:tot}
\epsilon ^2 \geq \epsilon (\varphi)^2 \geq \frac{1}{4} \frac{\mods{\langle [\Y(A)-A, L_{\Sy} + L_{\R}]\rangle};}{(\Delta L)^2}
\end{equation}
with perfect agreement coming with $\epsilon = 0$. Now, $[\Y(A),L] = 0$ and therefore,
since $(\Delta L)^2 = (\Delta_{\varphi} L_{\Sy})^2 + (\Delta_{\xi} L_{\R})^2$ and $\varphi$ is arbitrary, if $[A,L_{\Sy}] \neq 0$, there is a positive lower bound on the discrepancy between $A$ and $\Y(A)$ which can only be made smaller
by increasing the spread of $L_{\R}$ in $\xi$ \cite{ldlthesis}. This spread of course coincides with
a highly localised conjugate quantity, if one exists, for instance position conjugate to momentum,
angle conjugate to angular momentum, or even unsharp phase conjugate to number. 

Actually, the above results can be substantially improved and can be shown to hold in
operational terms -- see \cite{appr}. There it is also shown that a ``large" apparatus is needed
for good agreement between $A$ and $\Y(A)$ (for general effects $A$), and that poor localisation
gives poor approximation. Moreover, it is proven there that this type of behaviour persists
even when relative observables are not obtained via the $\Y$ map, i.e., good agreement between
an arbitrary quantity of $\Sy$ and an invariant one of $\Sy + \R$ requires large spread in the symmetry generator on $\hir$. This corresponds to high localisation in a conjugate quantity, should there be one.

This demonstrates that the reference frame, as a physical system,
is fulfilling its role appropriately in serving to define the absolute quantities of $\Sy$. That
badly localised states give bad approximation indicates that the approximation of relative by absolute is only as good as the reference system allows. This is most vividly born out in the case of position or angle, where the resulting absolute quantities $\Gamma_{\phi}(\cdot)$ after restriction are also not localisable.

If the reference is not sharply localised this is not a good reference for the ``absolute" position, as one might well expect. In the extreme situation that the reference state is completely delocalised (possible only in an approximate sense for position), the restricted quantity of $\Sy$ is shift-invariant, and thus commutes with momentum, and therefore says nothing about ``absolute" position at all. This situation is reminiscent of the apt metaphor in \cite{bj1}, where the concept of the
``edge of a mountain" is presented as one that is inherently vague, in order to facilitate the exposition of the idea of \emph{unsharp reality}, which may be seen as one of Paul Busch's major contributions to the philosophy of physics. The edge of a mountain, as intrinsically vague, does 
not function as a good reference for position. The position of a person relative to the edge of a mountain cannot be any more precisely specified than one can specify the edge of a mountain.
In quantum mechanics, as also elucidated in
\cite{bj1}, the vagueness arises from the \emph{nature} of the physical entities themselves, and not in deficiency of the concept, and therefore the question of the existence and properties of suitable references for 
quantum physical quantities is a deep issue.

The relative localisation of physical objects 
appears to be so ubiquitous in the classical world that it feels perhaps incommodious to view localisation itself as a relational property. However, this may be one of the essential differences between
classical and quantum physics, and we see the reference localisation (and also the possibility of networks of systems which are localised relative to each other) as something of a classicality condition -- a topic that we believe deserves further attention.

\subsection{Measurement}

We now return to the quantum theory of measurement, and the measurement of relational
quantities, specifically those of the form $(\Y \circ \Esf)(X)$ where $\Y$ is implicitly defined
for some $G$ and relativising quantity $\Fsf$. Measurements of $\Y \circ \Esf$ will then
be compared to measurements of $\Esf$, and the need of reference localisation for good approximation will again
be highlighted.

We therefore make the shift from viewing the system to be measured as $\Sy$ to $\Sy + \R$, i.e., let $\his \otimes \hir$
denote the Hilbert space of the system to be measured, and introduce the space $\hia$ to represent the apparatus. Therefore the total space is $\his \otimes \hir \otimes \hia $. Let $\tilde{\Esf} = \Y \circ \Esf: X \to L(\his \otimes \hir)$, which as an invariant quantity can be measured without constraint by a suitable pointer $Z$ on $\hia$. The probability reproducibility condition for the measurement of $\tilde{\Esf}$ can be written
\begin{equation}
\ip{\psi}{(\Y \circ \Esf)(X)\psi} = \ip{\psi}{\tilde{\Esf}(X)\psi} = \ip{U (\psi \otimes \phi)}{\id \otimes \Zsf (X) U (\psi \otimes \phi)},
\end{equation}
holding for all $X$ and unit vectors $\psi \in \his \otimes \hir$. Writing $\Zsf (X) (\tau) := U^* \id \otimes \Zsf (X) U$, this can be rewritten as 
\begin{equation}
(\Y \circ \Esf)(X)=\tilde{\Esf}(X) = \Gamma_{\phi}(\Zsf (X) (\tau)),
\end{equation} 
and $\tilde{\Esf}$ is the measured observable.

Now we want to see whether the same scheme can be re-purposed so as to function as 
a measurement scheme for $\Esf$. To do this, we fix some unit vector $\eta$ of the reference 
$\R$, which defines the observable $\Esf ^{(\eta)}$ of $\Sy$ through the formula
\begin{equation}
\ip{\varphi}{\Esf ^{(\eta)}(X) \varphi} = \ip{\varphi \otimes \eta}{\tilde{\Esf}(X)\varphi \otimes \eta}
\end{equation}
to hold, as always, for all unit $\varphi$ and $X$, which can also be written therefore at the operator level as $\Esf ^{(\eta)}(X)=\Gamma_{\eta}(\tilde{\Esf}(X))$. Then we may write
\begin{equation}
\ip{\varphi }{\Esf ^{(\eta)}(X) \varphi} = \ip{U(\varphi \otimes \eta \otimes \phi)}{\id \otimes \Zsf (X) U(\varphi \otimes \eta \otimes \phi)},
\end{equation}
where both $\phi$ and $\eta$ are fixed and may be viewed as part of the measuring apparatus.
This can be rewritten 
as $\Esf ^{(\eta)}(X) = \Gamma_{\eta \otimes \phi} (\Zsf (X) (\tau) = \Gamma_{\eta}(\Gamma_{\phi}(\Zsf (\tau) (X)))$.

Finally, we let $(\eta)_n$ be a localising sequence for $\Fsf$ at the origin (additive identity) of $G$, and take the high localisation limit, yielding 
\begin{equation}
\lim_{n \to \infty} \Gamma_{\eta _n}(\Gamma_{\phi}(\Zsf (\tau) (X))) = \Esf (X).
\end{equation}  

This observation highlights the central point of this investigation. The observable
$\tilde{\Esf}$ is a relational observable of $\Sy + \R$. This can be measured
whilst respecting the strong WAY theorem, and qualifies as a ``true observable" under
the principle that only invariants represent genuine physical quantities. However, under suitable reference preparation, arbitrary {\sc pom}s of $\Sy$, after the suppression of $\R$ (or the absorption of $\R$ into the apparatus) can function \emph{as though} they are observable,
even though they are not actually represented in physical reality. The localisation of the
reference state is the {\it sine qua non} of a well-functioning reference: it provides the
zero to which all else is then referred. What's more, high localisation with respect to
one quantity corresponds to large spread in a conjugate quantity, should one exist. High phase
localisation corresponds to large spread in number, high position localisation to large momentum spread, and so on.

The Ozawa position measurement scheme provides a good example of this behaviour. Instead of the absolute position, we instead use the same set-up to measure the relative position $Q \otimes \id - \id \otimes Q_{\mathcal{R}}\equiv Q-Q_\R$. Crucially, $P +P_{\R}$ is separately conserved and therefore falls under the remit of the strong WAY theorem.
The unique measured {\sc pom} in this scheme is $\widetilde{\mathsf{E}}: \mathcal{B} (\mathbb{R}) \to \mathcal{L}\bigl(\mathcal{H}_{\Sy}\otimes \mathcal{H}_{\R}\bigr) \equiv \mathcal{L}(L^{2} (\mathbb{R}^2)) $, extracted from the condition (writing $\Psi _{\tau} = U(\tau)\Psi_0$)
\begin{equation}\label{oz1}
\left\langle \Psi _{\tau }|\id\otimes \id\otimes 
\Zsf(f^{-1}(X)) \otimes \id \Psi _{\tau }\right\rangle
=\left\langle \varphi \otimes \phi| \widetilde{\mathsf{E}}(X)\varphi\otimes \phi \right\rangle,
\end{equation}
required to hold for all unit $\varphi \in \his$  and $\phi \in \hir$. It then follows that
\begin{equation}
\widetilde{\mathsf{E}}(X)=\chi _{X}\ast \widetilde{e}^{(\lambda )}(Q-Q_{\mathcal{R}}),
\end{equation}
where the right-hand side is the convolution of the set indicator function $\chi_X$ with the probability distribution $\widetilde{e}^{(\lambda )}(x)=\bigl\vert \xi_a^{(\lambda )}(x)\bigr\vert ^{2}$ with $\xi_a ^{(\lambda)}(s) = \sqrt{\lambda}\xi_a(\lambda s)$.

We see that $\widetilde{\mathsf{E}}$ is a smeared or unsharp version of the sharp $\Esf^{Q-Q_{\mathcal{R}}}$, and the quality of the approximation is again dictated by the spread of the density/confidence function $\widetilde{e}^{(\lambda )}$. The variance of $\widetilde{e}^{(\lambda )}$ is ${\var}(\widetilde{e}^{(\lambda )})=\frac{4}{\lambda ^{2}}{\var}\left\vert \xi_a\right\vert ^{2}$. Therefore by tuning $\lambda$ to be large, arbitrarily accurate measurements of $Q-Q_{\mathcal{R}}$ can be achieved, with no size/localisation requirements, and conclusion holds also for the overall width \cite{ldlthesis}.

It is only when trying to rehabilitate $Q$ as the observable-to-be-measured that the localisation
requirement reappears. By fixing $\phi$, the measurement scheme can be viewed as ``measuring" a {\sc pom} $\Esf$ for $\Sy$, in which we then return to the situation previously investigated:
$$
\ip{\varphi \otimes \phi}{\widetilde{\Esf}(X)\varphi \otimes \phi} =: \ip{\varphi}{\Esf (X) \varphi}.
$$ 
The probability distribution for the relative position
has thereby been re-expressed in terms of a ``smeared" distribution for the ``absolute" position by considering a fixed reference state $\phi$ of $\R$. 
The approximation error of $\Esf$ relative to $\Esf^Q$ is given by
$\var (e^{(\lambda)}) = \var \mods{\phi} + \frac{4}{\lambda ^2} \var \mods{\xi_a}$ (or, again, the overall width measure).
The probability distributions corresponding to the relative coordinate in the states $\varphi \otimes \phi$ become indistinguishable from those of the ``absolute" coordinate $Q$ in the limit that the localisation of the state $\phi$ with respect to $Q_{\mathcal{R}}$ 
is arbitrarily good (provided also that $\lambda$ is tuned to be large). 

Therefore we see that the limitation arises only when we attempt to ``measure" the non-invariant $Q$, and not the invariant $Q - Q_{\R}$ or, from the more WAY-type reasoning,
the measurement of $Q$, which does not commute with $P$, gives a limitation, but  $Q - Q_{\R}$ which does
commute with $P+P_{\R}$ does not, and this can be seen from within the measurement scheme of Ozawa.

The Ozawa model differs from the standard WAY set-up in that it includes the possibility of 
incorporating the reference as either part of the system or part of the apparatus. The standard WAY scenario, in our view, must be understood as pertaining to the situation wherein the single apparatus system must function also as the reference -- a situation to which we now turn.

\section{Relational view of the WAY theorem}\label{sec:relw}

We are now in position to outline the main conclusion of this work; first we briefly recap.
In section \ref{sec:way} we discussed the WAY theorem in its usual incarnation as presented in
\cite{LB2011,ldlthesis,lb2}, learning that 
observables (represented by self-adjoint operators) not commuting with the system part of an additive conserved quantity cannot be measured precisely in any measurement scheme satisfying the Yanase condition. However, we also saw that approximate measurements are possible, in the sense that the actually measured {\sc pom} becomes
statistically close to the ``target" sharp observable that we wish to measure, and that this comes
at the price of having to prepare an initial apparatus state with a large spread in the conserved quantity.

We then approached the question of characterising observability from a different perspective.
We provided two arguments regarding the imposition of symmetry on observability. The first was that fundamental symmetries limit what is observable to those quantities which are invariant under the symmetry transformation. This may be, for example, invariance under the action of the Galilei group or certain subgroups, the invariance then being a statement of the lack of absolute space, time and angle. 
The case of internal symmetries such as phase are bound by the same restriction, though the interpretation seems to be more directly linked to the fact that phase is inherently relative and is more akin to gauge. 

The second argument was derived from the quantum theory of measurement itself, and pertained to
systems which are isolated in the sense that they possess conserved quantities. Given this, from both points of view the conclusion is the same: no {\sc pom} not invariant under the symmetry action is (an) observable. We argued that the use of non-invariant quantities in the empirically
accurate description of quantum systems is a reflection of the fundamental point that such quantities are defined only with respect to a second system---a reference system---which can be externalised \emph{contingent} upon the preparation of a reference state highly localised 
in the appropriate variable.

Given that the compatibility condition  $[A,L_{\Sy}]=0$ in the context of the WAY theorem can be recast as the symmetry/invariance statement $A = e^{i \ell L_{\Sy}}Ae^{-i \ell L_{\Sy}}$, we observe that the conclusion of WAY, though conveying a similar moral to the point of view espoused in the subsequent section,
does not completely agree with it, because it does not rule out approximate measurements, i.e.,
unsharp approximators of $A$, described by a {\sc pom} $\E$ which is not compatible with $L_{\Sy}$.\footnote{The compatibility of a {\sc pom} $\Esf$ and a self-adjoint operator $M$ is equivalent
to the statement that $[\Esf(X),\Esf^M(Y)]=0$ for all $X$ and $Y$, where $\Esf^M(Y)$ is the spectral measure of $M$; see \cite{oqp}.} In our view, the situation encountered in the WAY theorem
is an expression of the general perspective that we have been advocating, and the possibility of approximate ``measurements" is actually an expression of the possibility of approximating relative observables by their unobservable/unmeasurable counterparts. We will now describe this in more detail.

\subsection{Measuring apparatus as reference frame}

We first recast the relevant compatibility conditions as invariance conditions, collecting together also some extra pieces of relevant information. Here the real parameter $\ell$ can be taken to lie in $\mathbb{R}$ or $S^1$. The physical setting is arranged for a measurement to take place, and thus we have systems $\Sy$ and $\Ap$, described together by the space $\his \otimes \hia$, together 
with a unitary coupling $U$, and so on. 
\begin{itemize}
\item Invariance of $\Sy$-observables: $[A,L_{\Sy}] =0$, if and only if $A = e^{i \ell L_{\Sy}}Ae^{-i \ell L_{\Sy}} \equiv \alpha^{\Sy}_{\ell}(A)$, if and only if $[\Esf^A(X), L_{\Sy}]=0$, if and only if $\alpha^{\Sy}_{\ell}(\Esf^A(X)) = \Esf^A(X)$ for all $X$. All of this generalises to the case where $\Esf^A$ is replaced by a {\sc pom} $\Esf$.
\item Conservation law for $\Sy + \Ap$: $[U,L_{\Sy} \otimes \id + \id \otimes L_{\Ap}] \equiv [U,L]=0$.
\item Yanase condition: $[Z,L_{\Ap}]=0$, which is equivalent to $Z = e^{i \ell L_{\Ap}}Ze^{-i \ell L_{\Ap}} \equiv \alpha^{\Ap}_{\ell}(Z)$, which under the conservation law $[U,L]=0$ is equivalent to $[U^*(\id \otimes Z)U,L] = 0$, or $[Z(\tau),L]=0$. This latter condition says that
the time evolved pointer $Z(\tau)$ must be invariant under $Z(\tau)\mapsto e^{i \ell L}Z(\tau)e^{-i \ell L}$.
\item \emph{Weak Yanase condition}: $[Z(\tau),L]=0$, i.e.,
the time evolved pointer $Z(\tau)$ must be invariant under $Z(\tau)\mapsto e^{i \ell L}Z(\tau)e^{-i \ell L}$ (this can be stipulated without the conservation or additivity of $L$).\footnote{That the WAY theorem admits a generalisation
in terms of the weak Yanase condition has been pointed out by Tukiainen in \cite{mikko}.}
\end{itemize}

The WAY theorem in its usual reading, with the Yanase condition in place of repeatability,
is ostensibly about the compatibility of $L_{\Ap}$ and $Z$ being ``inherited" by $A$ and $L_{\Sy}$.
Refashioning the compatibility conditions as invariance conditions allows for the impossibility part of the WAY theorem to be recast in 
the following (slightly informal) form: For any $Z$ satisfying the (weak) Yanase condition,
there is no measurement scheme which (perfectly) realises a measurement of (the self-adjoint operator) $A$ unless $A = \alpha^{\Sy}_{\ell}(A)$. Put yet another way, no non-invariant
$A$ can satisfy $A = \Gamma_{\phi} (Z(\tau))$. Given that $Z(\tau)$ is an invariant,
the connection with the discussion on reference frames becomes apparent. 

{\it A fortiori}, the ``positive" part of WAY asserts that $A$ and $\Gamma_{\phi} (Z(\tau))$
can become close, exactly when $\phi$ has a large $L_{\Ap}$--spread. In many instances of physical
interest, for instance when $L_{\Ap}$ is the momentum or angular momentum, this corresponds 
exactly to localisation in the conjugate variable (here, position or angle), and therefore
we see an extremely compelling link to the theory of quantum reference frames, where reference localisation was the key property for good approximation.


Now, we argue, that the WAY theorem as traditionally stated can be understood as pertaining
to the case wherein the apparatus $\Ap$ is functioning {\it both} as a measuring apparatus
in the traditional sense, \emph{and} as a reference system serving to define the relational observables which are traditionally treated as absolute. Indeed, there is no other system present which could serve this purpose -- in the standard set-up of WAY, there is no ``$\R$" system in addition to system and apparatus, but $\Sy + \Ap$ is certainly isolated, and the conservation of $L_{\Sy} + L_{\Ap}$ entails by strong WAY that no quantity of $\Sy + \Ap$ which does not commute with $L=L_{\Sy} + L_{\Ap}$/is not invariant under the symmetry action generated by $L$ can be measured ``from the outside" of $\Sy + \Ap$. This
is precisely the statement of the weak Yanase condition, i.e., that the time-evolved pointer $Z(\tau)$ is invariant. 

The localisation condition on the preparation of the reference frame corresponds exactly to large spread
in the conjugate quantity in the WAY setting, \emph{in the initial state of the apparatus}. 
There is therefore a strong harmony, both with respect to the impossibility part and the 
positive part. What is still unclear is the reference frames interpretation of WAY in the setting that the conserved quantity does not admit a conjugate quantity, in which case the large spread in the conserved quantity does not in any natural way correspond to a localised state with respect to another quantity. Understanding this remains as work to be done, although we note that almost all
physically interesting conserved quantities do admit conjugates.

\section{Discussion, future directions and closing Remarks}

Our conclusion, simply put, is that the Wigner-Araki-Yanase theorem, originally an expression of 
the limitations of the measurability of quantities not commuting with (the system part of) an additive conserved quantity, can be instead interpreted as a statement about the (im)possibilities of (unmeasurable) absolute quantities being good representations of observable relative quantities. 

The work presented comprises part of the vision to understand all aspects of quantum theory
from a relational viewpoint. There has already been much progress in this direction, and contemporary efforts are perhaps most strongly focused on understanding better the form of relationalism associated 
with quantum reference frames. By now the relational description of states and observables
is fairly well established, and the situation for quantum channels is described in a concurrent publication in this volume. The rejuvenation of foundational interest in quantum reference frames has spurred further recent progress in manifold directions. One major example arises from the 
need to describe transformations between quantum frames in the spirit of relativistic physics; see for example \cite{pal1,van1,flam1,gal1}. {\it Inter alia}, other recent work has provided
insights into the relative nature of coherence and superpositions \cite{long,zych1}, entanglement \cite{flam1}, the `Wigner's friend' thought experiment \cite{gal1}, and the nature of time in quantum mechanics \cite{time,smi1,smi2}. This latter topic   
is related in deep ways to the ``problem of time" in quantum gravity, and there are encouraging signs that the reference frames point of view may provide new insight there and in (quantum) gravitational physics more generally \cite{ho1,flam2,chat}.

The relational nature of quantum mechanics bears of course on the deep questions in quantum mechanics and beyond -- the distinction between subject and object, observer and observed or, to return to the elegant words of Heisenberg, between player and spectator. 
The most important question, however, takes us back to the beginning and to the WAY theorem, and is also the one that was closest to Paul's heart: what does all of this say about quantum measurement?

\end{document}